\documentclass[conference]{IEEEtran}

\usepackage{mathtools}
\usepackage{amssymb}
\usepackage{amsmath}
\usepackage{algorithm}
\usepackage{algpseudocode}
\usepackage{enumerate,multirow}
\usepackage{algpseudocode}
\usepackage{tikz}
\usepackage{tikz}
\usetikzlibrary{patterns,positioning,arrows,calc}
\usepackage{siunitx}
\usepackage{multirow}
\usepackage{multicol}
\usepackage{color}
\usepackage{xcolor}
\usepackage{algorithm}
\usepackage{algpseudocode}
\usepackage{graphicx}
\usepackage{subcaption}
\algrenewcommand\algorithmicrequire{\textbf{Input:}}
\algrenewcommand\algorithmicensure{\textbf{Output:}}
\usepackage{comment}

\usepackage{geometry}
\usepackage{tikz}
\usepackage{tikz}
\usetikzlibrary{patterns,positioning,arrows,calc}\usepackage{graphicx}
\usepackage{graphicx}
\usepackage{tikz}
\usetikzlibrary{trees}
\date{}

\usepackage{calc}
\newcolumntype{M}[1]{>{\centering\arraybackslash}m{#1}}
\newcolumntype{N}{@{}m{0pt}@{}}

\usepackage{mdwmath}
\usepackage{blindtext}
\usepackage{eqparbox}
\usepackage{stfloats}
\usepackage[numbers,sort&compress]{natbib}
\usepackage{algorithm}
\usepackage{algpseudocode}
\usepackage{tikz}
\usetikzlibrary{patterns, arrows.meta, positioning}
\IEEEoverridecommandlockouts
\title{\Huge
{List Decoding and New Bicycle Code Constructions for Quantum LDPC Codes}}
\author{\IEEEauthorblockN{Sheida Rabeti and Hessam Mahdavifar} 
\IEEEauthorblockA{Department of Electrical and Computer Engineering, Northeastern University, Boston, MA 02115, USA \\ 
Email: \{rabeti.s, h.mahdavifar\}@northeastern.edu}

}

\newtheorem{theorem}{{Theorem}}
\newtheorem{lemma}[theorem]{{Lemma}}
\newtheorem{proposition}[theorem]{{Proposition}}

\newtheorem{remark}{{Remark}}

\newcommand{\cC}{{\cal C}}

\newcommand{\cG}{{\cal G}}

\newcommand{\cL}{{\cal L}}

\newcommand{\cO}{{\cal O}} 
\newcommand{\cP}{{\cal P}}
 
\newcommand{\cR}{{\cal R}}
\newcommand{\cS}{{\cal S}} 
\newcommand{\cT}{{\cal T}}

\DeclareMathAlphabet{\mathbfsl}{OT1}{ppl}{b}{it} 

\newcommand{\be}[1]{\begin{equation}\label{#1}}
\newcommand{\ee}{\end{equation}}

\renewcommand{\le}{\leqslant} 
\renewcommand{\leq}{\leqslant}
\renewcommand{\ge}{\geqslant}

\newcommand{\Cref}[1]{Co\-ro\-lla\-ry\,\ref{#1}}

\usepackage{hvlogos}
\usepackage{qtree}
\usepackage{tikz}
\usepackage{multirow}
\usepackage{hhline}
\usepackage{hyperref}

\begin{document}

\vspace{10mm}
\maketitle

\begin{abstract}
In this paper, we propose a new decoder, called the Multiple-Bases Belief-Propagation List Decoder (MBBP-LD), for Quantum Low-Density Parity-Check (QLDPC) codes. It extends the Multiple-Bases Belief-Propagation (MBBP) framework, originally developed for classical cyclic LDPC codes. The proposed method preserves the linear-time complexity of standard BP decoder while improving the logical error rate. To further reduce the logical error rate, a new decision rule is introduced for the post-processing list decoder, outperforming the conventional least-metric selector (LMS) criterion.  For the recently developed and implemented bivariate bicycle (BB) code with parameters \([[144,12,12]]\), our proposed MBBP-LD decoder achieves up to 40\% lower logical error rate compared to the state-of-the-art decoder for short QLDPC codes, i.e., BP with ordered-statistics decoding (BP-OSD), while retaining the linear-time complexity of the plain BP decoder. In addition, we explore a new subclass of BB codes, that we refer to as the univariate bicycle (UB) codes, specifically with lower-weight parity checks (\(w=6,8\)). This reduces the polynomial search space for the code compared to general BB codes, i.e., by reducing the search space over two polynomial components in BB codes to just a single polynomial component in UB codes. Simulations demonstrate the promising performance of these codes under various types of BP decoders.
\end{abstract}

\section{Introduction}
Quantum error-correcting codes have emerged as one of the key enablers of quantum systems, providing the foundation for fault-tolerant quantum computation. In this context, classical coding-theoretic tools have proven to be highly effective. In particular, low-density parity-check (LDPC) codes, first introduced by Gallager \cite{gallager1962low}, have received significant attention for quantum computing in recent years. Their sparse parity-check matrices limit the number of required qubit–qubit interactions during error correction, making them especially suitable for fault-tolerant quantum architectures. Recent developments in quantum LDPC (QLDPC) constructions offer promising pathways toward high-performance quantum computing \cite{7336474,Breuckmann_2021,vasic2025quantumlowdensityparitycheckcodes}.

The state-of-the-art benchmark decoder for most QLDPC codes is BP-OSD, which combines belief propagation (BP) with a post-processing stage based on the order-statistics decoder (OSD) \cite{Panteleev_2021}. However, the complexity of BP-OSD is not favorable as it is dominated by the OSD stage, which in the worst case scales as $\cO(n^3)$, where $n$ is the code length. This limits their feasibility for practical systems as $n$ grows large, and motivates the search for alternative decoders.
To preserve the linear-time complexity of the plain BP decoder, several refinements such as layered decoding, serial scheduling, and bit-flipping strategies have been studied \cite{crest2023layereddecodingquantumldpc,Raveendran2021trappingsetsof,chytas2025enhancedminsumdecodingquantum}; however, their impact can be limited in the presence of dense short cycles and small trapping sets that often dominate QLDPC codes. Besides OSD, several other post-processing algorithms are studied for QLDPC codes such as stabilizer inactivation (BP-SI) \cite{crest2023stabilizerinactivationmessagepassingdecoding},and BP with guided decimation (BPGD) \cite{yao2024beliefpropagationdecodingquantum} with complexities of $\cO(n^2\log(n))$ and $\cO(n^2)$. They enhance the performance, although with the cost of super-linear time complexity.

To address these challenges, we propose and explore a new decoder, called the Multiple-Bases Belief-Propagation List-decoding (MBBP-LD) decoder, which introduces multiple redundant representations of the parity-check matrix and runs BP decoding on each in parallel. Each representation induces a distinct decoding trajectory, which helps to disrupt trapping sets and reduce the effect of short cycles when the redundancy is chosen carefully. The addition of redundant checks does not alter the code. Instead, the decoder is supplied with additional dual codewords, providing extra information that accelerates convergence. The resulting list of candidates is then combined through the decision-making function to produce the final output. 
This approach builds upon the Multiple-Bases Belief-Propagation (MBBP) framework that has been studied for classical LDPC codes \cite{Hehn_2010,4557244}, which we extend in multiple important respects to adapt it to QLDPC codes. This includes extending it by allowing identical rows in the redundant checks, employing explicit list decoding, and introducing structured partitions of the parity-check rows. In particular, we construct redundant parity checks from subtrees of the Tanner graph, yielding well-distributed and connected layerings that are demonstrated to improve convergence. With parallelization, the overall latency remains linear, as in the plain BP decoder, while even offering a potentially lower latency due to faster convergence. In addition, we propose new weight 6-8 generalized bicycle (GB) codes called Univariate Bicycle (UB) and study them under various BP decoders.

The rest of this paper is organized as follows. In Section II, some preliminaries on QLDPC codes are provided. In Section III, we present the proposed decoder framework, detailing its redundant parity-check layering and decision-making strategies. In Section IV, new UB code constructions are described. Further simulation results and performance comparisons are provided in Section V, followed by concluding remarks and future research directions in Section VI.
\section{Preliminaries}
\subsection{Quantum Stabilizer and CSS Codes}
An $[[n, k, d]]$ quantum stabilizer code is a $2^k$-dimensional subspace $\mathcal{C} \in (\mathbb{C}_2)^{\otimes n}$ with an Abelian stabilizer group 
$\mathcal{S}$:
\begin{equation} \label{eq:StabilizerCode} \mathcal{C} = \{|\psi\rangle \in (\mathbb{C}_2)^{\otimes n} \colon s |\psi\rangle= |\psi\rangle, \forall s \in \mathcal{S} \}. \end{equation}
Each generator $g \in \mathcal{S}$ acts as a parity-check constraint. The \textit{minimum distance} $d$ of a stabilizer code is the minimum weight of some Pauli operator $P \in \cP_n$ commuting with elements in $\cS$ such that $P \notin \cS$.
A Calderbank--Shor--Steane (CSS) code is a stabilizer code with a parity-check matrix of the form
\[
H = \begin{bmatrix} H_X & 0 \\ 0 & H_Z \end{bmatrix},
\]
where $H_X,H_Z$ are classical binary parity-check matrices satisfying $H_X H_Z^T = 0$. Such codes can correct Pauli-$X$ and Pauli-$Z$ errors independently using $H_Z$ and $H_X$, respectively. In this work, we focus on QLDPC CSS codes, where both $H_X$ and $H_Z$ are sparse.


\subsection{Belief Propagation (BP) Decoding}
BP decoding operates on the Tanner graph of $\cG = (V_c \cup V_v, E)$, where $V_c$ are check nodes, $V_v$ are variable nodes, and $E$ denotes edges. For each variable $v_i$, the log-likelihood ratio (LLR) is initialized as $
\mu_i = \log \frac{1-p}{p}$, where $p$ is the error probability of the physical qubits in $X$ and also the same probability in $Z$.
During iterations, messages are updated as:
\begin{align}
m^{(t+1)}_{c \to v} &= 2 \tanh^{-1} \!\!\prod_{v' \in N(c)\setminus v} \tanh\!\Big(\tfrac{1}{2} m^{(t)}_{v' \to c}\Big), \\
m^{(t+1)}_{v \to c} &= \mu_i + \sum_{c' \in N(v)\setminus c} m^{(t)}_{c' \to v}.
\end{align}
The posterior LLR is $
m^{(t)}_i = \mu_i + \sum_{c \in N(v_i)} m^{(t)}_{c \to v_i}$.
A hard decision is made as $\hat{e}_i = 0$ if $m^{(t)}_i \ge 0$, otherwise, $\hat{e}_i=1$. The algorithm halts when all parity checks satisfy the syndrome provided by the measurement (i.e., a valid error pattern is found) or when a maximum number of iterations is reached.


\subsection{Generalized Bicycle (GB) Codes}
A code is called \emph{cyclic} if it is closed under cyclic shifts. For any cyclic code $\cC$, we can associate a one-to-one mapping between $\mathbb{F}^n$ and $\cR_n \triangleq \mathbb{F}[x] / x^n-1$ by mapping $c = (c_{0}, c_1, \ldots, c_{n-1}) \in \mathbb{F}^n$ to $c(x) \triangleq c_0 + c_1x + ... + c_{n-1}x^{n-1}$ and, consequently, a cyclic shift corresponds to $xc(x) = c_{n-1} + c_0x + ... + c_{n-2}x^{n-1}$. 
Hence, every cyclic code forms an ideal in $\mathcal{R}$ generated by a monic polynomial 
$g(x)$ where $g(x)\,|\,x^n - 1$, with check polynomial $h(x) = (x^n - 1)/g(x)$. 
Let $P$ be the $n \times n$ cyclic permutation matrix. Then, both generator and parity-check matrices can be expressed as circulant matrices where a circulant matrix $A$ corresponding to the polynomial $a(x) = a_0 + a_1x+ ... + a_{n-1}x^{n-1}$ is defined as $A = a(P)$.

Since circulant matrices commute, they are very useful for CSS code constructions, e.g., they are used in \cite{PhysRevA.88.012311} to construct bicycle codes. 
The \emph{generalized bicycle} (GB) construction \cite{PhysRevA.88.012311} defines a CSS quantum LDPC code using two $n \times n$ matrices $A$ and $B$, typically circulant or quasi-circulant:
\begin{equation}
H_X = [A,\,B], \qquad H_Z = [B^{\mathsf{T}},\,A^{\mathsf{T}}],
\label{eq:gb}
\end{equation}
with $H_X H_Z^{\mathsf{T}} = AB + BA = 0$, ensuring CSS orthogonality. When $B=A^{\mathsf{T}}$, this reduces to the \emph{bicycle} codes \cite{MacKay_2004}.

Consider two circulant matrices $A$ and $B$ corresponding to the polynomials $a(x),b(x)\in\mathbb{F}_2[x]$ respectively, with degree $<n$.  
Then the code $\mathrm{GB}(a,b)$ with dimension $k$ corresponding to CSS $[[2n, k ]]$ code is given by the following proposition.

\begin{proposition}
\label{prop1}
The dimension $k$ of the generalized bicycle code $[[2n, k]]$ defined by 
$a(x), b(x) \in \mathbb{F}_2[x]$ is given by:
\begin{equation}
    k = 2 \deg h(x)
\end{equation}
where $h(x) \triangleq \gcd(a(x), b(x), x^n - 1)$.
\end{proposition}

\section{Multiple-Bases Belief-Propagation List-decoding (MBBP-LD)}
\label{sec:RPC-LD}
In this section, we recall the Multiple-Bases Belief-Propagation (MBBP) framework \cite{Hehn_2010,4557244} and extend it by introducing a structured method to generate redundant parity-check matrices for QLDPC codes. We also propose an improved decision rule that enhances decoding performance compared to the conventional least-metric selector (LMS).

\subsection{MBBP Decoding via Tree-Based Construction}
\label{sec:RPC-LD-A}
MBBP decoding runs multiple BP decoders in parallel, each operating on a distinct parity-check matrix representation of the same code. Let the matrix used by the $\ell$-th decoder be denoted by $H^{(\ell)}$, $\ell \in \{1, \ldots, L\}$, and the error vector found after at most $i$ iterations by $\hat{e}_\ell$. Each instance of the decoder that has converged contributes its output to a candidate list $\mathcal{L} = \{\hat{e}_s \,|\, s \in \mathcal{S}\}$, where $\mathcal{S} \subseteq \{1, \ldots, L\}$ denotes the indices of successful decoders. In prior work \cite{Hehn_2010,4557244}, the elements of $\mathcal{L}$ are then passed to a least metric selector (LMS), which selects the most likely codeword according to the channel distribution. In this work, we further introduce an alternative decision-making rule that will be described in the next subsection.

We obtain the parity-check matrices corresponding to the parallel decoders by extending the original matrix $H$ with redundant layers derived from a collection $\cT$ of subtrees in the Tanner graph that partition the check nodes. A subtree $t \in \cT$ is called \emph{maximal} if no additional check node, together with its adjacent variable nodes, can be included without forming a cycle.
The set $\mathcal{T}$ is determined in an ad-hoc fashion by exploring the check nodes and forming maximal sub-trees. The construction procedure is detailed in Algorithm \ref{alg:subtree}.
Each subtree $t$ defines a local submatrix $H_t$, and the corresponding representation of the code is given by
\begin{equation}
    H^{(t)} = 
    \begin{bmatrix}
        H \\[2pt]
        H_t
    \end{bmatrix}.
\end{equation}
Subsequently, BP decoding is run in parallel on all constructed matrices $\{H^{(t)} \,|\, t \in \mathcal{T}\}$. Each decoder produces an estimate $\hat{e}_t$ after a fixed number of iterations, and those that converge successfully add their outputs to the candidate list $\mathcal{L} = \{\hat{e}_t \,|\, t \in \mathcal{T}_{\mathrm{conv}}\}$, where $\mathcal{T}_{\mathrm{conv}} \subseteq \mathcal{T}$ denotes the set of successful decoders. The list $\mathcal{L}$ is then passed to a decision-making function $f_{\mathrm{DM}}$, which selects the final estimate $\hat{e}$ according to a certain selection rule, to be specified later in the next subsection. The redundant-row MBBP decoding scheme is illustrated in Fig.~\ref{fig:tree-mbbp} with pseudo-codes provided in  Algorithm~\ref{alg:mbbp-decoder}.
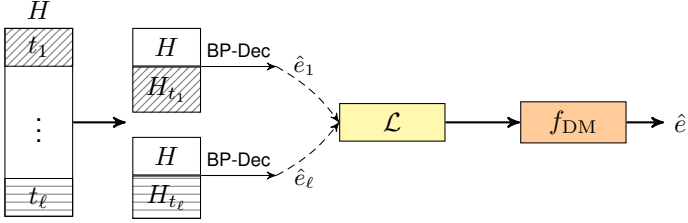
\begin{figure}[!t]
\centering
\begin{tikzpicture}[>=stealth', font=\sffamily]

\node[draw, minimum width=0.9cm, minimum height=2.5cm, label=above:$H$] (H) {};

\node[draw, pattern=north east lines, pattern color=gray,
      minimum width=0.9cm, minimum height=0.5cm,
      anchor=north] (S1) at (H.north) {$t_1$};

\node at (H.center) {$\vdots$};

\node[draw, pattern=horizontal lines, pattern color=gray,
      minimum width=0.9cm, minimum height=0.5cm,
      anchor=south] (Sl) at (H.south) {$t_\ell$};

\draw[->, thick] (H.east) -- ++(0.7,0);
\coordinate (Hmid) at (H.center);

\node[draw, minimum width=0.9cm, minimum height=0.5cm,
      right=1.7cm of Hmid, yshift=1cm, anchor=center] (H1) {$H$};
\node[draw, pattern=north east lines, pattern color=gray,
      minimum width=0.9cm, minimum height=0.6cm,
      below=0cm of H1.south, anchor=north] (HS1) {$H_{t_1}$};

\draw[->] (HS1.north east) -- ++(1.0,0) node[midway, above] {\scriptsize BP-Dec}
           coordinate (bp1);

\node[draw, minimum width=0.9cm, minimum height=0.5cm,
      right=1.7cm of Hmid, yshift=-0.45cm, anchor=center] (H2) {$H$};
\node[draw, pattern=horizontal lines, pattern color=gray,
      minimum width=0.9cm, minimum height=0.3cm,
      below=0cm of H2.south, anchor=north] (HSl) {$H_{t_\ell}$};

\draw[->] (HSl.north east) -- ++(1.0,0) node[midway, above] {\scriptsize BP-Dec}
           coordinate (bpl);

\node[draw, fill=yellow!40, minimum width=1.4cm, minimum height=0.5cm,
      right=4.7cm of Hmid, anchor=center] (L) {$\mathcal{L}$};
\node[draw, fill=orange!40, minimum width=1.4cm, minimum height=0.5cm,
      right=1.7cm of L, anchor=center] (DM) {$f_{\mathrm{DM}}$};

\draw[->, thick] (L.east) -- (DM.west);
\draw[->, thick] (DM.east) -- ++(0.5,0) node[right] {$\hat{e}$};

\draw[densely dashed, ->] (bp1) to[bend left=12] (L.west);
\draw[densely dashed, ->] (bpl) to[bend right=12] (L.west);

\node[font=\small] at ($(bp1)!0.45!(L.west)+(0,0.35)$) {$\hat{e}_1$};
\node[font=\small] at ($(bpl)!0.45!(L.west)-(0,0.35)$) {$\hat{e}_\ell$};

\end{tikzpicture}
\caption{MBBP-LD decoding with redundant-row construction.}
\label{fig:tree-mbbp}
\end{figure}
\begin{algorithm}[!t]
\caption{MBBP-LD Decoder via Decision-Maker $f_{\mathrm{DM}}$ for QLDPC Codes}
\label{alg:mbbp-decoder}
\begin{algorithmic}[1]
\Require Parity-check matrix $H \in \{0,1\}^{m \times n}$; collection of maximal subtrees $\mathcal{T} = \{t_1, \ldots, t_{|\mathcal{T}|}\}$; syndrome vector $s \in \{0,1\}^m$; channel parameter $p$; decision rule $f_{\mathrm{DM}}$
\Ensure Estimated error vector $\hat{e}$

\State Initialize candidate list $\mathcal{L} \gets \emptyset$

\Statex \textbf{Decoding Phase:}
\For{each $t \in \mathcal{T}$}
    \State $H^{(t)} \gets [\,H; H_t\,]$
    \State $(\hat{e}_t, \textit{converged}) \gets \textsc{BP-Decode}(H^{(t)}, s, p)$
    \If{\textit{converged}}
        \State $\mathcal{L} \gets \mathcal{L} \cup \{\hat{e}_t\}$
    \EndIf
\EndFor

\Statex \textbf{Decision Phase:}
\State $\hat{e} \gets f_{\mathrm{DM}}(\mathcal{L})$

\Return $\hat{e}$
\end{algorithmic}
\end{algorithm}
\begin{algorithm}[!t]
\caption{Maximal Subtrees Construction}
\label{alg:subtree}
\begin{algorithmic}[1]
\Require Parity-check matrix $H \in \{0,1\}^{m \times n}$; permutation $\pi$ of check nodes $V_c$
\Ensure Collection of maximal subtrees $\mathcal{T} = \{t_1, \ldots, t_{|\mathcal{T}|}\}$
\State Construct the Tanner graph $G = (V_c \cup V_v, E)$.
\State Mark all $c \in V_c$ as unvisited; set $\mathcal{T} \gets \emptyset$.
\For{each $c \in V_c$ in order $\pi$}
    \If{$c$ unvisited}
        \State Initialize subtree $t \gets \{c\}$, mark visited and queue $Q \gets \{c\}$.
        \While{$Q \neq \emptyset$}
            \State Remove $u$ from $Q$.
            \For{each $c' \in V_c$ adjacent to $u$ via one variable node}
                \If{$c'$ unvisited and $t \cup \{c'\}$ is cycle-free}
                    \State Add $t \gets \{c'\}$, mark visited, and queue $Q \gets \{c'\}$.
                \EndIf
            \EndFor
        \EndWhile
        \State Append $\mathcal{T} \gets \{t\}$.
    \EndIf
\EndFor
\State \Return $\mathcal{T}$
\end{algorithmic}
\end{algorithm}
In essence, the procedure performs a breadth-first traversal of the Tanner graph, where at each iteration all variable-node neighbors of the current check node are included to expand the subtree. Different permutations $\pi$ of the check nodes lead to different traversal orders, producing different collections $\mathcal{T}$ and consequently varied redundant matrices $H^{(t)}$, which potentially enhances decoding performance for different codes.
\begin{lemma}[Subtree Size Bound]
\label{lem:subtree-size}
For a Tanner graph $\mathcal{G} = (V_c \cup V_v, E)$ with check-regular degree $w$, 
the number of check nodes in any subtree $t \in \mathcal{T}$ generated by Algorithm  \ref{alg:subtree} satisfies
\begin{equation}
    |t| \le \frac{|V_v| - 1}{w - 1}.
\end{equation}
In particular, for $w = 6$ generalized bicycle (GB) codes where $|V_v| = 2|V_c|$, this bound reduces to
\begin{equation}
    |t| \le \frac{2|V_c| - 1}{5} \le 0.4\,|V_c|.
\end{equation}
\end{lemma}

\begin{proof}
Each check node in a tree of degree $w$ is connected to $w$ distinct variable nodes. 
When a new check is added to the subtree, it must share exactly one variable node 
with the existing check nodes to maintain acyclicity; hence, each additional check 
introduces exactly $(w-1)$ new variable nodes. 

Consequently, a subtree containing $|t|$ check nodes covers $1 + (w-1)|t|$ variable nodes considering the check node root as well. 
Since the Tanner graph contains at most $|V_v|$ variable nodes, it follows that
$1 + (w-1)|t| \le |V_v|$, which simplifies to the desired bound 
$|t| \le (|V_v| - 1)/(w - 1)$. 
\end{proof}

 Lemma \ref{lem:subtree-size} provides a bound on the complexity for cases where the decoder also depends on the number of parity checks. We discuss this specifically for BP with serial scheduling in Section \ref{sec:num-res}.
\begin{remark}
Since the parity-check matrices are sparse and QLDPC codes typically have moderate block lengths, the partitioning step incurs negligible computational cost in practice. 
The procedure consists of successive breadth-first searches over the Tanner graph, with total complexity $\mathcal{O}(|E|)$, linear in the number of edges. 
This cost is insignificant compared to the iterative BP decoding process and therefore does not affect the overall complexity of the proposed decoder. 
\end{remark} 

\begin{remark} 
This framework preserves maximum-likelihood (ML) decoding, as each subtree $t \in \cT$ contains all variable nodes adjacent to its check nodes, ensuring that no additional constraints are introduced. Redundant parity checks introduce additional dual codewords (stabilizer combinations) that enhance belief-propagation convergence. By running redundant decoders in parallel, the scheme retains the linear-time complexity and latency of standard BP. Moreover, a threshold $\tau$ can be defined such that the decoding process terminates once the fraction of converged processes reaches $\tau$.
\end{remark}

Intuitively, the subtree-based construction enhances decoding by mitigating trapping sets \cite{Raveendran2021trappingsetsof, 1003839} and helping with well-distributed layering. 
Since BP decoding is exact on tree-like graphs \cite{gallager1962low}, it remains exact on the submatrix $H_t$. The induced subtrees also yield balanced and connected partitions in the Tanner graph, especially in graphs with short cycles. This follows from the girth properties and similar arguments in \cite{Rabeti2025BoundsAN}. Such structured and evenly distributed layerings improve convergence compared to random selection, enabling the redundant matrices $H^{(t)}$ to enhance decoding performance.

\subsection{Decision Making (DM) Rule}
After constructing the candidate list $\cL$, as discussed in Section\,\ref{sec:RPC-LD-A}, the output error vector is given by $\hat{e} = f_{\mathrm{DM}}(\cL)$. We propose the following rule:


\textbf{Frequency-Weighted Scoring (FWS).}  
    Each candidate is assigned a score that reflects its frequency in the list and its Hamming weight:
    \[
        f_{\mathrm{DM}}^{\text{FWS}}(\cL) 
        = \arg\max_{e \in \cL} \frac{|\{e' \in \cL : e'=e\}|}{w_H(e)+1}.
    \]
   The numerator rewards candidates that are repeatedly produced by different BP instances, while the denominator penalizes higher weight errors.
    The addition of one in the denominator avoids division by zero when $w_H(e)=0$.

If no decoder converges, an extra BP stage with a higher iteration limit can be applied, or non-converged outputs may be considered as candidates. 
In the reported simulations, such cases are treated as decoding failures, and the all-zero error vector is returned.

\section{Univariate Bicycle (UB) Code}
\label{sec:construct}
In this section, we recall the Frobenius identity and use it to propose a new construction called \textit{Univariate Bicycle (UB)} codes. These codes are derived from Generalized Bicycle (GB) codes \cite{PhysRevA.88.012311} with row weight limited to $w$. This method reduces the polynomial search space from $\cO(n^{w})$ to $\cO(n^{w/2})$, i.e., instead of searching for two polynomial components $a(x)$ and $b(x)$, it only searches for $a(x)$ and sets $b(x)$ to a carefully chosen power of $a(x)$. This allows us to obtain codes with parameters close to those in \cite{Bravyi_2024, wang2024coprime, postema2025existencecharacterisationbivariatebicycle}.
The Frobenius Identity is as follows:
\begin{proposition}[Frobenius Identity]
\label{lem:Frobenius}
Let $p$ be a prime number, and let $\mathbb{F}$ be a field of characteristic $p$.  
Then, for any elements $x_1, x_2, \ldots, x_n \in \mathbb{F}$,
\begin{equation}
    (x_1 + x_2 + \cdots + x_n)^t 
    = x_1^t + x_2^t + \cdots + x_n^t.
\end{equation}
where $t = p^k$, for some integer $k \ge 1$.
\end{proposition}
Our construction uses the Frobenius identity to preserve the weight-limited structure of the parity-check matrices $H_x$ and $H_z$, maintaining row weights of 6, 8. 
Let $A$ be a circulant matrix generated by $a(x)$. 
By Proposition \ref{lem:Frobenius}, for any $\ell>0$, the polynomial $b(x) \triangleq a^{t}(x)$ with $t = 2^{\ell}$, preserves the Hamming weight of $a(x)$. 
If $\gcd(a(x), x^n - 1) = g(x)$ with $\deg(g(x)) = k$, then by Proposition \ref{prop1} the corresponding code $\mathrm{GB}(a,b)$ has dimension $2k$ and parameters $[[2n,\,2k,\,d]]$, with stabilizer generators of weight $w=2\,w_H(a)$. 
While BB codes \cite{Bravyi_2024} and their coprime version \cite{postema2025existencecharacterisationbivariatebicycle, wang2024coprime} use bivariate polynomials for greater flexibility, our UB construction corresponds to the univariate case ($m=1$ in \cite{Bravyi_2024}), achieving comparable parameters with a much smaller search complexity. Selected codes can be seen in Table \ref{tab:poly-params}

\begin{algorithm}[!t]
\caption{Univariate Bicycle (UB) Code Search}
\label{alg:UB-search}
\begin{algorithmic}[1]
\Require Code length $n$, target row weight $w$, and power limit $\ell_{\max}$
\Ensure A list of candidate $(a(x), b(x))$ pairs generating UB codes
\State Initialize $\mathcal{L} \gets \emptyset$ \Comment{List of valid code pairs}
\For{each polynomial $a(x) \in \mathbb{F}_2[x]$ with $w_H(a)=w$ and $\deg(a)<n$}
        \For{$\ell = 1, 2, \ldots, \ell_{\max}$}
            \State $b(x) \gets a(x)^{2^{\ell}} \bmod (x^n + 1)$
                \State Construct circulant matrices $A = \mathrm{circ}(a(x))$, $B = \mathrm{circ}(b(x))$
                \State Form parity-check matrices $H_X = [A \, B]$, $H_Z = [B^{\!\top} \, A^{\!\top}]$
                \If{$\dim(\mathrm{CSS}(H_X, H_Z)) > 2$}
                    \State Append $(a(x), b(x))$ to $\mathcal{L}$
                \EndIf
        \EndFor
\EndFor
\State \Return $\mathcal{L}$
\end{algorithmic}
\end{algorithm}

\textbf{Complexity Analysis:}
The Algorithm \ref{alg:UB-search} determines the suitable polynomial $a(x)$, whose search space has an order of $\mathcal{O}(n^{w/2})$, 
while $b(x)$ is directly obtained from $a(x)$ using the exponent parameter $\ell$. 
In contrast, optimizing the general BB code construction requires an exhaustive search over both $a(x)$ and $b(x)$, 
resulting in a total complexity of $\mathcal{O}(n^{w})$.

This structure allowed faster searches for both weight-8 and weight-6 codes, leading to several codes with new parameters. The codes $[[126,12,\leq10]]$ and $[[126,14,\leq10]]$ have similar parameters to the coprime code $[[126,12,10]]$ \cite{wang2024coprime}, but show slight improvement in simulations. Table \ref{tab:poly-params} lists selected codes obtained by Algorithm \ref{alg:UB-search}. Exhaustive search was used to determine exact and lower-bound distances, while linear programming estimated upper bounds in infeasible cases.

\begin{table}[!t]
\centering

\label{tab:poly-params}
\renewcommand{\arraystretch}{1.0}
\setlength{\tabcolsep}{3.4pt}
\small
\begin{tabular}{|c|>{\centering\arraybackslash}m{1.2cm}|c|c|c|}
\hline
$a(x)$ & $\ell$ & $[[n,k,d]]$ & $R = \frac{k}{n}$ & $w$ \\
\hline
$1+x^{2}+x^{3}+x^{6}$ & 2 & $[[124,12,\ge10]]$ & 0.096 & 8 \\
\hline
$1+x+x^{4}+x^{7}$ & 3 & $[[124,14,\leq 10]]$ & 0.113 & 8 \\
\hline
$1+x^{2}+x^{3}+x^{9}$ & 3 & $[[126,14,\leq 10]]$ & 0.111 & 8 \\
\hline
$1+x^{2}+x^{5}+x^{6}$ & 4 & $[[126,12, 10]]$ & 0.095 & 8 \\
\hline
$1+x+x^{12}+x^{16}$ & 2 & $[[126,12,\le10]]$ & 0.095 & 8 \\
\hline
$1+x+x^{6}$ & 3 & $[[126,12, 8]]$ & 0.095 & 6 \\
\hline
$x^{4}+x^{3}+x+1$ & 3 & $[[132,8,8]]$ & 0.060 & 8 \\
\hline
$x^{8}+x^{7}+x+1$ & 2 & $[[140,16,8]]$ & 0.114 & 8 \\
\hline
$x^{7}+x^{4}+x^{3}+1$ & 2 & $[[144,14,8]]$ & 0.097 & 8 \\
\hline
$x^{10}+x^{9}+x^{2}+1$ & 4 & $[[146,20,8]]$ & 0.137 & 8 \\
\hline
$x^{10}+x^{8}+x+1$ & 4 & $[[146,20,8]]$ & 0.137 & 8 \\
\hline
$x^{8}+x^{7}+x^{5}+1$ & 2 & $[[168,16,\ge8]]$ & 0.095 & 8 \\
\hline
$x^{9}+x^{4}+x+1$ & 2 & $[[168,18,\ge8]]$ & 0.107 & 8 \\
\hline
$x^{12}+x^{10}+x^{9}+1$ & 2 & $[[178,24,\ge8]]$ & 0.135 & 8 \\
\hline
$x^{6}+x^{5}+x+1$ & 5 & $[[180,12,\ge8]]$ & 0.067 & 8 \\
\hline
$x^{7}+x^{3}+x+1$ & 2 & $[[180,14,\ge8]]$ & 0.078 & 8 \\
\hline
$x^{8}+x^{6}+1$ & 9 & $[[180,16,\ge8]]$ & 0.089 & 6 \\
\hline
$x^{19}+x^{6}+x+1$ & 2 & $[[312,14,\le13]]$ & 0.045 & 8 \\
\hline
$x^{12}+x^{4}+1$ & 2 & $[[560, 24,\le8]]$ & 0.044 & 8 \\
\hline
\end{tabular}
\caption{Selected univariate bicycle (UB) codes with $b(x) = a(x)^t$, $t = 2^\ell$.}
\label{tab:poly-params}
\vspace{-5mm}
\end{table}
\section{Numerical Results}
\label{sec:num-res}
\subsection{Proposed Decoder Performance Comparison}
In Fig.  \ref{fig:decoder_comparison}, simulation results are shown for two different BB codes [[144, 12, 12]] and [[288, 12, 18]] \cite{Bravyi_2024} over a binary symmetric channel with independent \(X\)-type errors. Sampling was terminated once 100 decoding failures were observed. We also compared MBBP-LD with other decoders, namely, Parallel BP, Serial BP, and BP-OSD. In Fig.  \ref{fig:bp_144}, all decoders have been set to $I_{max} = 600$, with BP using parallel normalized min-sum (with normalization factor $\beta = 0.875$), except the Serial BP decoder, which uses serial scheduling. The BP-OSD has order 0. 
In Fig.  \ref{fig:bp_288}, all decoders except BP are set to $I_{max} = 1000$, all using the serial min-sum decoder ($\beta = 0$), except the parallel BP, which uses parallel scheduling with $I_{max} = 50000$. The BP-OSD has order 0. In both cases, it can be seen that across the full range \(p \in [0.03,0.10]\), the proposed decoder consistently outperforms the other decoders while keeping the linear-time complexity.
The BP, BP-Serial, and BP-OSD decoders were implemented using the libraries in \cite{Roffe_2020, Roffe_LDPC_Python_tools_2022}. 
\begin{figure}[!t]
    \centering
    \begin{subfigure}[t]{0.48\textwidth}
        \centering
        \includegraphics[width=\linewidth]{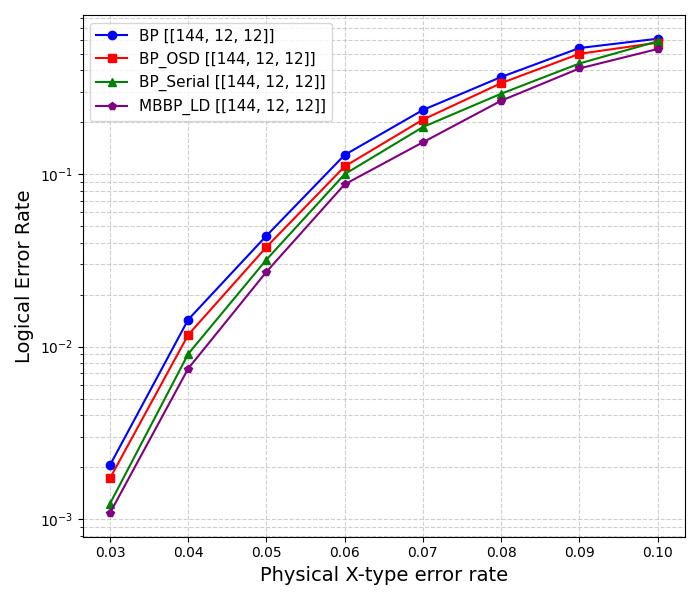}
        \caption{$\tau = 1, I_{\max}=600, \beta = 0.875, [[144,12,12]]$ \cite{Bravyi_2024}}
        \label{fig:bp_144}
    \end{subfigure}
    \begin{subfigure}[t]{0.48\textwidth}
        \centering
        \includegraphics[width=\linewidth]{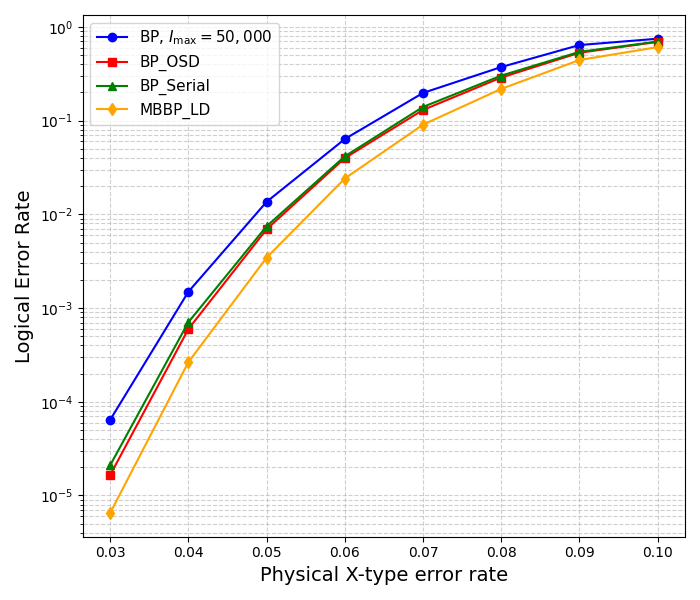}
        \caption{$\tau = 1, I_{\max}=1000, \beta = 0, [[288,12,18]]$ \cite{Bravyi_2024}}
        \label{fig:bp_288}
    \end{subfigure}
\caption{Performance of the proposed MBBP-LD decoder under two different regimes (Fig. \ref{fig:bp_144} and Fig. \ref{fig:bp_288} have parallel and serial scheduling, respectively)}
    \label{fig:decoder_comparison}
    \vspace{-5mm}
\end{figure}
\subsection{Complexity Analysis of MBBP-LD}
As MBBP-LD directly uses the BP decoder without post-processing, its complexity under parallel scheduling is $\mathcal{O}(n)$. By Lemma \ref{lem:subtree-size}, for a weight-$w$ code, each redundant parity-check matrix $H_t$ contains at most $\gamma = \frac{n-1}{w-1}$ check nodes. Hence, the decoding time satisfies  
\[
T_{\mathrm{MBBP-LD}} = \mathcal{O}((1+\gamma) \, T_{\mathrm{BP-Serial}}),
\]
which gives $T_{\mathrm{MBBP-LD}} = \mathcal{O}(1.4 \, T_{\mathrm{BP-Serial}})$ for $w = 6$.  
We compared the runtime of MBBP-LD with BP, BP-Serial, and BP-OSD decoders. For fairness, all decoders were tuned to achieve a similar logical error rate (LER). BP used parallel scheduling with $I_{\max} = 50000$, while other decoders used serial scheduling with $I_{\max} = 100$. All decoders used the min-sum algorithm with $\beta = 0$; BP-OSD had order $10$, and MBBP-LD used a stopping threshold $\tau = 0.4$.

It is worth noting that serial scheduling often improves convergence due to its sequential message updates, which can lead to faster stabilization of beliefs despite higher per-iteration complexity. Table \ref{tab:runtime_xonly} 
presents the simulation results for error rates $\{0.02, 0.06, 0.1\}$. The results show that MBBP-LD, under parallel scheduling, achieves nearly the same latency as linear-time decoders, even for short block lengths and low error-rate regimes. We also observed no significant performance difference when reducing the stopping threshold $\tau$, which can be tuned based on the code and operating error-rate regime. All simulations were performed on a 15-inch MacBook Air (Apple M3, 8 GB RAM, 2024) using Python 3.11, with parallelization implemented via \texttt{ThreadPoolExecutor}.
\begin{table}[t]
\centering
\renewcommand{\arraystretch}{1.1}
\resizebox{\linewidth}{!}{
\begin{tabular}{lcccc}
\hline
\textbf{Decoder} & \textbf{$p_x$} & \textbf{Avg. runtime (ms)} & \textbf{LER $P_L^X$} & \textbf{$T_{\mathrm{BP}}/T_{\mathrm{decoder}}$} \\
\hline
BP         & 0.02 & 0.01944 & $1.11\times10^{-4}$ & 1.00$\times$ \\
BP\_Serial & 0.02 & 0.01853 & $1.20\times10^{-4}$ & 1.05$\times$ \\
BP\_OSD10    & 0.02 & 0.02093 & $1.10\times10^{-4}$ & 0.93$\times$ \\
MBBP\_LD   & 0.02 & 0.02002 & $9.40\times10^{-5}$ & 0.97$\times$ \\
\hline
BP         & 0.06 & 9.1609  & 0.0970              & 1.00$\times$ \\
BP\_Serial & 0.06 & 0.09925 & 0.1074              & 92.30$\times$ \\
BP\_OSD10    & 0.06 & 0.12022 & 0.08682             & 76.20$\times$ \\
MBBP\_LD   & 0.06 & 0.08519 & 0.07279             & 107.53$\times$ \\
\hline
BP         & 0.10 & 69.09   & 0.6130              & 1.00$\times$ \\
BP\_Serial & 0.10 & 0.425   & 0.6365              & 162.6$\times$ \\
BP\_OSD10    & 0.10 & 0.544   & 0.5678              & 126.9$\times$ \\
MBBP\_LD   & 0.10 & 0.323   & 0.5443              & 213.7$\times$ \\
\hline
\end{tabular}}
\caption{Runtime of decoders for the \([[144,12,12]]\) code 
with stopping threshold $\tau = 0.4$ and $I_{\max} = 100$ 
(except BP with $I_{\max} = 50{,}000$).}
\label{tab:runtime_xonly}
\vspace{-4mm}
\end{table}
\subsection{Performance of the Proposed UB Codes}
Fig.  \ref{fig:ub_comparison} shows the logical error rate of two proposed UB codes, $[[126,12,\leq10]]$ and $[[126,14,\leq10]]$. Under the depolarizing error model, $X$, $Y$, and $Z$ errors occur independently with equal probability $q$. To compare with the $X$-type error model, with error probability $p$, the parameters satisfy $p = 2q/3$. These codes are compared with the coprime BB code $[[126,12,10]]$ \cite{wang2024coprime}. The BP-OSD and MBBP-LD decoders use the min-sum algorithm with serial scheduling $(\beta = 0)$. To illustrate the effect of the stopping threshold in MBBP-LD, we set $\tau = 0.4$. The OSD order is 7, and $I_{max} = 1000$. Both UB codes achieve slightly better performance than the coprime version, with further gains observed under the MBBP-LD decoder, even though the UB $[[126,14,\leq10]]$ code has a higher dimension.
\begin{figure}[!t]
\centering
\includegraphics[width=\linewidth]{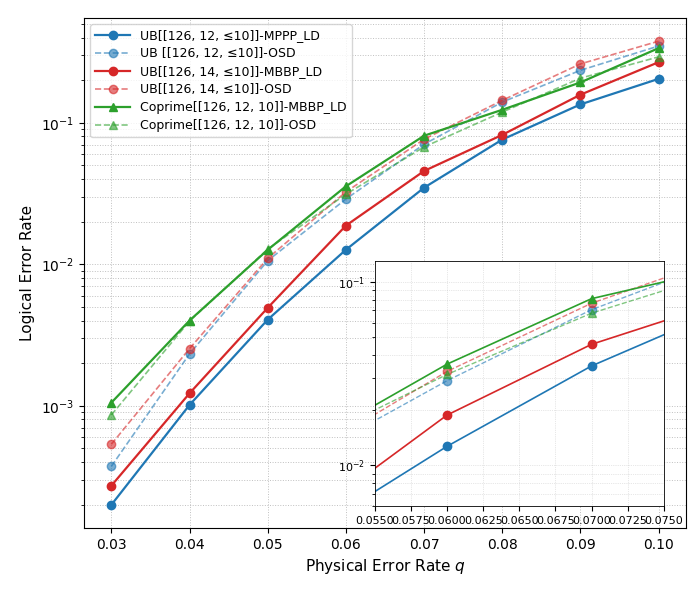}
\caption{Performance of the proposed UB codes with $(I_{max}=1000, \beta=0, \tau=0.4)$.}
\label{fig:ub_comparison}
\vspace{-5mm}
\end{figure}

\vspace{-2mm}
\section{Conclusion}
\label{sec:conclusion}
In this paper, we proposed MBBP-LD decoder for QLDPC codes, which can be applied to both cyclic and non-cyclic codes. Numerical results have shown consistent improvements compared to the state-of-the-art BP-OSD in a wide range of operating regimes.
We also proposed and studied UB codes that offer a reduced search space for the code design compared to BB codes. 
There are several directions for future research. For instance, we expect that making the construction of redundant checks specific to each code family, designing stronger decision-making rules, optimizing the stopping threshold, and adapting BP parameters such as maximum iterations and normalization factors can further enhance the performance. Moreover, the framework naturally benefits from trying different permutations and layerings, offering decoding diversity with little added latency when executed in parallel. For the proposed UB codes, it is natural to consider larger lengths and expand the choice of code parameters, where a brute-force search for optimal BB codes become increasingly difficult. Any advancement along these directions can further strengthen the applicability of QLDPC codes to quantum computing systems by reducing the decoding complexity/latency, improving the error rates, and offering new tools for more efficient search for good QLDPC codes with enhanced parameters.   

\bibliographystyle{IEEEtran}
{\footnotesize \bibliography{reffff}}

\end{document}